\documentclass[review]{elsarticle}

\usepackage{amsmath,amsthm,amssymb,latexsym,stmaryrd}
\usepackage{tikz}
\usepackage{tikz-cd}
\usepackage[mathscr]{euscript} 
\usepackage{stackrel}
\usepackage[all]{xy}

\journal{Journal of \LaTeX\ Templates}

\begin{document}

\begin{frontmatter}

\title{The Theory of an Arbitrary Higher $\lambda$-Model \tnoteref{mytitlenote}}


\author{Daniel O. Martínez-Rivillas}

\author{Ruy J.G.B. de Queiroz}


\begin{abstract}
One takes advantage of some basic properties of every homotopic $\lambda$-model (e.g.\ extensional Kan complex) to explore the higher $\beta\eta$-conversions, which would correspond to proofs of equality between terms of a theory of equality of any extensional Kan complex. Besides, Identity types based on computational paths are adapted to a type-free theory with higher $\lambda$-terms, whose equality rules would be contained in the theory of any $\lambda$-homotopic model.
\end{abstract}


\end{frontmatter}


\newtheorem{defin}{Definition}[section]
\newtheorem{teor}{Theorem}[section]
\newtheorem{corol}{Corollary}[section]
\newtheorem{prop}{Proposition}[section]
\newtheorem{rem}{Remark}[section]
\newtheorem{lem}{Lemma}[section]
\newtheorem{nota}{Notation}[section]
\newtheorem{ejem}{Example}[section]

\section{Introduction}

In \cite{Martinez} and \cite{MartinezHoDT} the initiative is born to search for higher $\lambda$-models with non-trivial structure of $\infty$-groupoid, by using extensional Kan  complexes $K\simeq[K\rightarrow K]$. In \cite{Martinez21} the existence of higher non-trivial models is proved by solving homotopy domain equations.

\medskip If we understand an arbitrary higher $\lambda$-model as an extensional  Kan complex, the following question arises: What would be the syntactic structure of the equality theory of any higher $\lambda$-model, i.e., is its equality theory a generalization of the $\beta\eta$-conversions to $(n)\beta\eta$-conversions in a set $\Lambda_{n-1}(a,b)$ by  $(n)\beta\eta$-contractions induced by the extensionality from a Kan complex?.

\medskip We shall see some consequences of the equality theory $Th(\mathcal{K})$ of an extensional Kan complex $\mathcal{K}$ with some examples of equality and nonequality of terms. This paves the way for a definition of the $(n)\beta\eta$-conversions, which will belong to the set of $n$-conversions $\Lambda_n$ induced by the least theory of equality on all the extensional Kan complexes, here called \emph{Homotopy Type-Free Theory} (\emph{HoTFT}).

\medskip On the other hand, we define, from the identity types based on computational paths \cite{Queiroz}, the untyped theory of higher $\lambda\beta\eta$-equality TH-$\lambda\beta\eta$. We ask about the relationship between  TH-$\lambda\beta\eta$ and HoTFT.

\medskip In this work we will try to answer these questions according to the following sections:
In section 2, we explore the theory of any extensional Kan complex in order to generalize the $\beta\eta$-conversions to $(n)\beta\eta$-conversions in a set $\Lambda_{n-1}(a,b)$ by  $(n)\beta\eta$-contractions induced by the extensionality from a Kan complex. In section 3, the identity types $Id_A(a,b)$ based on computational paths are taken into account, to define a type-free theory of higher $\lambda\beta\eta$-equality TH-$\lambda\beta\eta$ with  $\lambda^n$-terms  and $n$-redexes in a set $\Lambda^{n-1}(a,b)$ with $n\geq 1$. Finally, we look at the relationship of this TH-$\lambda\beta\eta$  with the least theory of equality on all the extensional Kan complexes HoTFT through the relationship between the sets $\Lambda^n$ and $\Lambda_n$ for each $n\geq 0$. 

\section{Theory of extensional Kan complexes}

In this section, we shall see some consequences of the equality theory $Th(\mathcal{K})$ of an extensional Kan complex $\mathcal{K}$ with some examples of equality and nonequality of terms. This shall pave the way for a definition of the $(n)\beta\eta$-conversions, which will belong to the set of $n$-conversions $\Lambda_n$ induced by the least theory of equality on all extensional Kan complexes, denoted by HoTFT.

\begin{defin}[$\infty$-category \cite{DBLP:books/mk/Lurie}]
	An $\infty$-category is a simplicial set $X$ which has the following property: for any $0<i<n$, any map $f_0:\Lambda_i^n\rightarrow X$ admits an
	extension $f:\Delta^n\rightarrow X$.	
\end{defin}
\noindent Here the simplicial set $K$ is defined as a presheaf $\Delta^{op}\rightarrow Set$, with $\Delta$ being the \textit{simplicial indexing category}, whose objects are finite ordinals $[n]=\{0,1,\ldots,n\}$, and morphisms are the (non strictly) order preserving maps. $\Delta^n$ is the \textit{standard $n$-simplex} defined for each $n\geq 0$ as the simplicial set $\Delta^n:=\Delta(-,[n])$. And $\Lambda_i^n$ is a \textit{horn} defined as largest subobject of $\Delta^n$ that
does not include the face opposing the $i$-th vertex.

\begin{defin}
	From the definition above, we have the following special cases:
	\begin{itemize}
		\item $X$ is a Kan complex if there is an extension for each $0\leq i\leq n$.
		\item $X$ is a category if the extension exists uniquely \cite{DBLP:books/mk/Rezk22}.
		\item $X$ is a groupoid if the extension exists for all $0\leq i\leq n$ and is unique \cite{DBLP:books/mk/Rezk22}. 
	\end{itemize}
\end{defin}

In other words,  a Kan complex is an $\infty$-groupoid; composed of objects, $1$-morphisms, $2$-morphisms, ...,  all those invertible.

\begin{nota}
	For $K$ a Kan simplex and $n\geq 0$, let $K_n=Fun(\Delta^n,K)$ be the Kan complex of the $n$-simplexes.
	
	Let $Var$ be the set of all variables of $\lambda$-calculus, for all $m,n\geq 0$, each assignment  $\rho:Var\rightarrow K_n$ ($\rho(t)$ is an $n$-simplex of $K$, for each $t\in Var$), $x\in Var$ and $f\in K_m$, denote by $[f/x]\rho$  the assignment $\rho':Var\rightarrow K$ which coincides with $\rho$, except on $x$, where $\rho'$ takes the value $f$.
\end{nota}

\begin{defin}[h.p.o \cite{MartinezHoDT}]\label{h.p.o}
	Let $\hat{K}$ be an $\infty$-category. The largest Kan complex $K\subseteq\hat{K}$ is a homotopy partial order (h.p.o), if for every $x,y\in K$ one has that $\hat{K}(x,y)$ is contractible or empty. Hence, the Kan complex $K$ admits a relation of h.p.o $\precsim$ defined for each $x,y\in K$ as follows:
	$x\precsim y$ if $\hat{K}(x,y)\neq\emptyset$, hence the pair $(K,\precsim)$ is a h.p.o. (we denote simply by $K$). 
\end{defin}

\begin{defin}[c.h.p.o \cite{MartinezHoDT}] 
	Let $K$ be an h.p.o.
	\begin{enumerate}
		\item An h.p.o $X\subseteq K$ is directed if $X\neq \emptyset$ and for each $x,y\in X$, there exists $z\in X$ such that $x\precsim z$ and $y\precsim z$.
		\item $K$ is a complete homotopy partial order (c.h.p.o) if
		\begin{enumerate}
			\item There are initial objects, i.e.,  $\bot\in K$ is a initial object if for each $x\in K$, $\bot\precsim x$. 
			\item For each directed $X\subseteq\mathcal{K}$ the supremum (or colimit) $\bigcurlyvee X\in\mathcal{K}$ exists. 
		\end{enumerate}
	\end{enumerate}
\end{defin}

\begin{defin}[Continuity \cite{Martinez21}]
	Let $K$ and $K'$  be c.h.p.o's. A functor $F:K\rightarrow K'$ is continuous if $F(\bigcurlyvee X)\simeq\bigcurlyvee F(X)$, where $F(X)$ is the essential image.  
\end{defin}

\begin{defin}[$CHPO$ \cite{Martinez21}]
	Define the subcategory $CHPO\subseteq CAT_\infty$ whose objects are the c.h.p.o's and the morphisms are the continuous functors, where $CAT_\infty$ is the $\infty$-category of the $\infty$-categories \cite{DBLP:books/mk/Lurie}.
\end{defin}

\begin{defin}[Reflexive Kan complex\footnote{In \cite{MartinezHoDT} one can also see the relationship between the reflexive Kan complexes and syntactic homotopic $\lambda$-models, conceptually introduced in \cite{Martinez}, analogously to the semantics of the classic $\lambda$-calculus; same for the relationship between complete partial orders (c.p.o.'s) and syntactic $\lambda$-models.} \cite{MartinezHoDT}]
	A quadruple $\langle K,F,G,\varepsilon\rangle$ is called a reflexive Kan complex, if $K$ is a c.h.p.o such that the full subcategory  $[K\rightarrow K]\subseteq Fun(K,K)$ of the continuous functors is a retract of $K$, via the functors
	$$F:K\rightarrow[K\rightarrow K],\hspace{1cm} G:[K\rightarrow K]\rightarrow K$$  
and the natural equivalence $\varepsilon:FG\rightarrow 1_{[K\rightarrow K]}$. If  there is a natural equivalence $\eta:1_K\rightarrow GF$, the quintuple $\langle K,F,G,\varepsilon,\eta\rangle$ represents an extensional Kan complex.	
\end{defin}

Just as the recursive Domain Equation $X\cong [X\rightarrow X]$ (in the category of the c.p.o's) has an implicit recursive definition of data-types, the ``Homotopy Domain Equation" \cite{Martinez21} $X\simeq [X\rightarrow X]$ (in the $\infty$-category $CHPO$)  would also have a recursive definition of  data-types. A recursively defined computational object (e.g., a proof by mathematical induction) would be of a higher order relative to the classical case, whose interpretation would be recursively defined by a sequence of partial functors $F_i:K\rightarrow K$, over a  Kan complex $K$ weakly ordered, which converges to a total functor $F:K\rightarrow K$, whose details are not among the objectives of this work, but will be developed in future works, when studying the semantics (case of inductive types) of the version of HoTT based on computational paths.

\begin{ejem}[\cite{Martinez21}]
	The c.h.p.o $K_\infty$, which generalizes Dana Scott's c.p.o $D_\infty$, is an extensional Kan complex, since $K_\infty$ is a solution for the Homotopy Domain Equation $X\simeq[X\rightarrow X]$ in the $\infty$-category $\,CHPO$ of c.h.p.o's and continuous functors.
\end{ejem} 

Thus, intuitively, from the computational point of view, we have that a Kan complex, which satisfies the Homotopy Domain Equation, is not only capable of verifying the computability of constructions typical of classical programming languages, as $D_\infty$ does it, but it also has the advantage (over $D_\infty$) of verifying the computability of higher constructions, such as a mathematical proof of some proposition, the proof of the equivalence between two proofs of the same proposition, etc.

\bigskip Besides, in \cite{Martinez21}, several examples of extensional objects (Kan complexes) are presented in the Kleisli $\infty$-category $Kl(P)$.

\begin{defin}[\cite{MartinezHoDT}]\label{interpretation-definition}
	Let $K$ be a reflexive Kan complex (via the morphisms $F$, $G$).
	\begin{enumerate}
	\item For $f,g:\triangle^n\rightarrow K$ (or also $f,g\in K_n$) define the $n$-simplex
	$$f\bullet_{\triangle^n} g=F(f)(g).$$
	In particular for vertices $a,b\in K$,
	$$a\bullet b=a\bullet_{\triangle^0} b=F(a)(b),$$
	besides, $F(a)\bullet(-)=a\bullet(-)$ and $F(-)(b)=(-)\bullet b$ are functors on $K$, then for $f\in K_n$ one defines the $n$-simplexes 
	$$a\bullet f=F(a)(f),\hspace{1cm}f\bullet b=F(f)(b).$$
		
		\item For each $n\geq 0$, let $\rho$ be a valuation at $K_n$. Define the interpretation $\llbracket \,\,\, \rrbracket_\rho:\Lambda\rightarrow K_n$ by induction as follows
		\begin{enumerate}
			\item $\llbracket x \rrbracket_\rho=\rho(x),$
			\item $\llbracket MN \rrbracket_\rho=\llbracket M \rrbracket_\rho\bullet\llbracket N \rrbracket_\rho$,
			\item $\llbracket \lambda x.M \rrbracket_\rho= G(\boldsymbol{\lambda} f.\llbracket M \rrbracket_{[f/x]\rho})$, where $\boldsymbol{\lambda} f.\llbracket M \rrbracket_{[f/x]\rho}=\llbracket M \rrbracket_{[-/x]\rho}:K\rightarrow K_n$. 
		\end{enumerate}
	\end{enumerate}
\end{defin}

\begin{rem}\label{Remark-1}
	Given  $g\in K_n$ and $\rho:Var\rightarrow K_n$, the higher $\beta$-contraction is interpreted by
	\begin{align*}
	\llbracket \lambda x.M \rrbracket_\rho\bullet g&=G(\boldsymbol{\lambda} f.\llbracket M \rrbracket_{[f/x]\rho})\bullet g \\
	&=F(G(\boldsymbol{\lambda} f.\llbracket M \rrbracket_{[f/x]\rho}))(g) \\
	&\xrightarrow{(\varepsilon_{\boldsymbol{\lambda} f.\llbracket M \rrbracket_{[f/x]\rho}})_g} (\boldsymbol{\lambda} f.\llbracket M \rrbracket_{[f/x]\rho}) (g) \\
	&=\llbracket M \rrbracket_{[g/x]\rho},
	\end{align*}
	where  $\varepsilon_{\boldsymbol{\lambda} f.\llbracket M \rrbracket_{[f/x]\rho}}$ is the natural equivalence, induced by $\varepsilon$, between the functors $F(G(\boldsymbol{\lambda} f.\llbracket M \rrbracket_{[f/x]\rho}),\boldsymbol{\lambda} f.\llbracket M \rrbracket_{[f/x]\rho}:K\rightarrow K_n$. Hence $(\varepsilon_{\boldsymbol{\lambda} f.\llbracket M \rrbracket_{[f/x]\rho}})_g$ is the equivalence induced by the $n$-simplex $g$ in $K$. 
	
	  \medskip Hence, if $\langle K, F, G, \varepsilon,\eta\rangle$ is extensional and $n=0$, so that the $\beta$-contraction is modelled by $\varepsilon:FG\rightarrow 1$; the (reverse) $\eta$-contraction  is modelled by  $\eta:1\rightarrow GF$. Besides, if $n>0$, we have that the natural equivalences $\varepsilon$ and $\eta$ will induce higher $\beta$-contractions and (reverse) $\eta$-contractions  respectively, as we will see later.
\end{rem}

\begin{prop}\label{proposition1}
	Let $x,y,M,N,P$ be $\lambda$-terms. The interpretations of $\beta$-reductions 
	\[\xymatrix{
		& {(\lambda x. M) ((\lambda y. N) P)}\ar[d]_{1\beta}\ar[r]^{1\beta} &
		[(\lambda y. N)  P/x]M \ar[d]^{[1\beta]}\\
		& (\lambda x. M) ([P/y]N)\ar[r]_{1\beta} & [([P/y]N)  /x]M
	}\]
	are equivalent in every reflexive Kan complex $\langle K,F,G,\varepsilon\rangle$. 
\end{prop}
\begin{proof}
	Let $a=\llbracket P \rrbracket_\rho$,  $\llbracket \lambda y.N \rrbracket_\rho\bullet a\xrightarrow{f}\llbracket N \rrbracket_{[a/y]\rho}$,  $R=FG(\boldsymbol{\lambda} f.\llbracket M \rrbracket_{[f/x]\rho})$, $L=\boldsymbol{\lambda} f.\llbracket M \rrbracket_{[f/x]\rho}$ and $\varepsilon'=\varepsilon_{\boldsymbol{\lambda} f.\llbracket M \rrbracket_{[f/x]\rho}}$. One has that the natural equivalence $\varepsilon':R\rightarrow L$ makes the following diagram (weakly) commute:
	\[\xymatrix{
		& {R(\llbracket \lambda y. N \rrbracket_{\rho}\bullet a)}\ar[d]_{R(f)}\ar[r]^{\varepsilon'_{\llbracket \lambda y. N \rrbracket_{\rho}\bullet a}} &
		\,\,L(\llbracket\lambda y. N \rrbracket_{\rho}\bullet a) \ar[d]^{L(f)}\\
		& R(\llbracket N \rrbracket_{[a/y]\rho})\ar[r]_{\varepsilon'_{\llbracket N \rrbracket_{[a/y]\rho}}} & L(\llbracket N \rrbracket_{[a/y]\rho})
		& 
	}\]
	which, by  Remark \ref{Remark-1}, corresponds to the (weakly)  commutative diagram
	\[\xymatrix{
		& {\llbracket\lambda x. M\rrbracket_\rho\bullet (\llbracket\lambda y. N\rrbracket_\rho\bullet a)}\ar[d]_{R(f)}\ar[r]^{\hspace{0.5cm}\varepsilon'_{\llbracket \lambda y. N \rrbracket_{\rho}\bullet a}} &
		\,\, \llbracket M \rrbracket_{[\llbracket\lambda y. N\rrbracket_\rho\bullet a/x]} \ar[d]^{L(f)}\\
		& \llbracket\lambda x. M\rrbracket_\rho\bullet \llbracket N \rrbracket_{[a/y]\rho}\ar[r]_{\hspace{0.5cm}\varepsilon'_{\llbracket N \rrbracket_{[a/y]\rho}}} & \llbracket M\rrbracket_{[\llbracket N \rrbracket_{[a/y]\rho}/x]}
		& 
	}\] 	
\end{proof}

\begin{ejem}
	The $\lambda$-term $(\lambda x.u)((\lambda y.v)z)$ has two $\beta$-reductions:
	\[\xymatrix{
		& {(\lambda x. u )((\lambda y. v) z)}\ar[d]_{1\beta}\ar[r]^{1\beta} &
		[(\lambda y. v)  z/x]u \ar[d]^{[1\beta]}\\
		& (\lambda x. u) ([z/y]v)\ar[r]_{1\beta} & [v/x]u
	}\]
	making $u=M$, $v=N$ and $z=P$, by Proposition \ref{proposition1}, the interpretations of these $\beta$-reductions are equivalent in all reflexive Kan complexes $\langle K,F,G,\varepsilon\rangle$. 
\end{ejem}

Next, we shall give examples where the reductions of $\lambda$-terms are not equivalent.

\begin{ejem}
	The $\lambda$-term $(\lambda x.(\lambda y.yx)z)v$ has the $\beta$-reductions
	\[\xymatrix{
		& {(\lambda x.(\lambda y.yx)z)v}\ar[d]^{1\beta}\ar[r]^{1\beta} &
		(\lambda y.yv)z \ar[d]^{1\beta}\\
		& (\lambda x.zx)v\ar[r]^{1\beta} & zv
	}\]
	Given a reflexive Kan complex $\langle K,F,G,\varepsilon\rangle$. Let $\rho(v)=c$, $\rho(z)=d$ vertices at $K$ and  $R=FG$. The interpretation of the $\beta$-reductions of $(\lambda x.(\lambda y.yx)z)v$ depends on solving the diagram equation
	\[\xymatrix{
		& {R(\boldsymbol{\lambda} a.R(\boldsymbol{\lambda} b.b\bullet a)(d))(c)}\ar[d]_{(R(?))_c}\ar[r]^{\hspace{0.6cm}(\varepsilon_f)_c} & R(\boldsymbol{\lambda} b.b\bullet c)(d) \ar[d]^{(\varepsilon_{\boldsymbol{\lambda}b.b\bullet c})_d}\\
		& R(\boldsymbol{\lambda} a.d\bullet a)(c)\ar[r]_{(\varepsilon_g)_c} & d\bullet c
	}\]
where $f=\boldsymbol{\lambda}a.R(\boldsymbol{\lambda} b.b\bullet a)(d)$ and $g=\boldsymbol{\lambda}a.d\bullet a$ are functors at $[K\rightarrow K]$. One has  $h_a=(\varepsilon_{\boldsymbol{\lambda}b.b\bullet a})_d:f(a)\rightarrow g(a)$ for each vertex $a\in K$, but $h_a$ is not necessarily a functorial equivalence in any reflexive Kan complex $\langle K,F,G,\varepsilon\rangle$ to get the diagram  to commute:
	\[\xymatrix{
	& {R(f)(c)}\ar[d]_{(R(h?))_c}\ar[r]^{(\varepsilon_f)_c} & f(c) \ar[d]^{h_c}\\
	& R(g)(c)\ar[r]_{(\varepsilon_g)_c} & g(c)
}\]
\end{ejem}

\begin{ejem}
	The $\lambda$-term $(\lambda z.xz)y$ has the $\beta\eta$-contractions 
	\[\xymatrix{
		& {(\lambda z.xz)y}\ar@/^2mm/[r]^{\,\,\,1\beta} \ar@/_2mm/[r]_{\,\,\,1\eta} & xy
	}\]
	
	Take an extensional Kan complex $\langle K,F,G,\varepsilon,\eta\rangle$. Let $\rho(x)=a$ and $\rho(y)=b$ be vertices of $K$. The interpretation of $\lambda$-term is given by: $\llbracket(\lambda z.xz)y \rrbracket_\rho=\llbracket\lambda z.xz \rrbracket_\rho\bullet b=G(\boldsymbol{\lambda} c.F(a)(c))\bullet b=G(F(a))\bullet b=(FGF)(a)(b)$. The interpretation of the $\beta\eta$-contractions corresponds to the degenerated diagrams 
	\[\xymatrix{
		& (FGF)(a)(b)\ar[d]_{(\varepsilon_{F(a)})_b}\ar[dr]^{1} & &F(a)(b)\ar[d]_{(F(\eta_a))_b}\ar[dr]^{1}\\
		&F(a)(b)\ar[r]_{(F(\eta_a))_b} & (FGF)(a)(b) & (FGF)(a)(b)\ar[r]_{\hspace{0.4cm}(\varepsilon_{F(a)})_b}  & F(a)(b)
	}\]
	But the diagrams do not necessarily commute in every extensional Kan complex $\langle K,F,G,\varepsilon,\eta\rangle$. 
\end{ejem}
For examples of higher extensional $\lambda$-models see \cite{Martinez21}.

\medskip It is known that the types of HoTT correspond to $\infty$-groupoids. Taking advantage of this situation, for a reflexive Kan complex, let us define the theory of equality on that Kan complex ($\infty$-groupoid) as follows.

\begin{defin}[Theory of an extensional Kan complex]\label{Definition-Theory-Kan-Complex}
	Let $\mathcal{K}=\langle K,F,G,\varepsilon,\eta\rangle$ be an extensional Kan complex. Define the theory of equality of $\mathcal{K}$ as the class
	$$Th_1(\mathcal{K})=\{M=N\,| \, \llbracket M \rrbracket_\rho
	\simeq\llbracket N \rrbracket_\rho \,\, \text{for all}\,\,\rho:Var\rightarrow K\}$$
	where $\llbracket M \rrbracket_\rho\simeq\llbracket N \rrbracket_\rho$ is the equivalence between vertices of $K$ for some equivalence $\llbracket s \rrbracket_\rho:\llbracket M \rrbracket_\rho\rightarrow\llbracket N \rrbracket_\rho$, and ``$s$" denotes the conversion between $\lambda$-terms $M$ and $N$ induced by $\llbracket s \rrbracket_\rho$ for all evaluation $\rho$.
\end{defin}

In the Definition \ref{Definition-Theory-Kan-Complex}, notice that the equivalence $\llbracket M \rrbracket_\rho\simeq\llbracket N \rrbracket_\rho$ for all $\rho$, induces the intentional equality $M=N$, which can be seen as an identity type based on computational paths \cite{Queiroz}; the conversion $s$ may also be seen as a computational proof (a finite sequence of basic rewrites \cite{Queiroz} induced by $K$) of the proposition $M=N$ in the theory $Th_1(\mathcal{K})$.

\begin{rem}
	If $s$ is a $\beta$-contraction or $\eta$-contraction and the functor $F$ is not surjective for objects, the equality $M=_{1\beta} N:K$ or $M=_{1\eta} N:K$ is not necessarily a judgmental equality (as it happens in HoTT); $\llbracket M \rrbracket_\rho$ and $\llbracket N \rrbracket_\rho$ may be different vertices in $K$. Thus,   the theory $Th_1(\mathcal{K})$ may be seen as the family of all the identity types which are inhabited by paths which are not necessarily equal to the reflexive path $\mathtt{refl}_M$. 
\end{rem} 

\begin{nota}\label{nota1}
	Let $M$ and $N$ be $\lambda$-terms ($M,N\in\Lambda_0$) and $\mathcal{K}$ be an extensional Kan complex. Denote by $\Lambda_0(K)(M,N)$ the set of all the $1$-conversions from $M$ to $N$ induced by $\mathcal{K}$. We write $\Lambda_1(\mathcal{K}):=\bigcup_{M,N\in\Lambda_0}\Lambda_0(\mathcal{K})(M,N)$  for the family of all $1$-conversions induced by $\mathcal{K}$.
	
	Let $s,t\in\Lambda_0(\mathcal{K})(M,N)$. Denote by $\Lambda_0(\mathcal{K})(M,N)(s,t)$  the set of all the $2$-conversions from $s$ to $t$. And let $\Lambda_2(\mathcal{K}):=\bigcup_{s,t\in\Lambda_1}\bigcup_{M,N\in\Lambda_0}\Lambda_0(\mathcal{K})(M,N)(s,t)$  be the family of all $2$-conversions induced by $\mathcal{K}$, and so on we keep iterating for the families $\Lambda_3(\mathcal{K})$, $\Lambda_4(\mathcal{K}),\ldots$.
\end{nota}

Since $\mathcal{K}$ is a reflexive Kan complex,  $Th_1(\mathcal{K})$ is an intentional $\lambda$-theory of $1$-equality which contains the theory $\lambda\beta\eta$. Iterate again, we have the $\lambda$-theory of $2$-equality  
	$$Th_2(\mathcal{K})=\{r=s\,\,| \,\, \forall\rho\,(\llbracket r \rrbracket_\rho\simeq\llbracket s \rrbracket_\rho)\text{\hspace{0.1cm} and\hspace{0.2cm}} r,s\in\Lambda_0(\mathcal{K})(M,N)\}.$$ 
If we keep iterating, we can see that the reflexive Kan complex $\mathcal{K}$ will certainly induce a $\lambda$-theory of higher equality given by the inverse and direct limit 
$$Th(\mathcal{K})=\bigcup_{n\geq 1}Th_n(\mathcal{K}).$$

Just as $Th_1(\mathcal{K})$ contains $\lambda\beta\eta$, $Th(\mathcal{K})$ will contain a (simple version of) `Homotopy Type-Free Theory', defined as follows.  

\begin{defin}[Homotopy Type-Free Theory]\label{Definition-Homotopy-Type-Free-Theory}
	A Homotopy Type-Free Theory (HoTFT)  consists of the least theory of equality, that is
	$$HoTFT:=\bigcap\{Th(\mathcal{K})\,|\,\mathcal{K} \,\text{is an extensional Kan complex}\}.$$
	And for each $n\geq 0$ let $$\Lambda_n:=\bigcap\{\Lambda_n(\mathcal{K})\,|\,\mathcal{K} \,\text{is an extensional Kan complex}\}$$
	be the set of $n\beta\eta$-conversions.
\end{defin}

For example, let $\mathcal{K}=\langle K,F,G,\varepsilon,\eta\rangle$ be an extensional Kan complex and $x$, $M$ and $N$ $\lambda$-terms. By Definition \ref{Definition-Homotopy-Type-Free-Theory}, the $\beta$-contraction $(\lambda x.M)N\xrightarrow{1\beta}[N/x]M$ inhabits the set $\Lambda_0((\lambda x.M)N,[N/x]M)$; 
$$\llbracket 1\beta \rrbracket_\rho=(\varepsilon_{\llbracket M \rrbracket_{[-/x]\rho}})_{\llbracket N \rrbracket_\rho}\in K(\llbracket (\lambda x.M)N \rrbracket_\rho,\llbracket [N/x]M \rrbracket_\rho),$$ 
and the $\eta$-contraction $\lambda x.Mx\xrightarrow{1\eta}M$, $x\notin FV(M)$, belongs to $\Lambda_0(K)(\lambda x.Mx,M)$; 
$$\llbracket 1\eta \rrbracket_\rho=\eta_{\llbracket M \rrbracket_{\rho}}\in K(\llbracket\lambda x.Mx \rrbracket_\rho,\llbracket M \rrbracket_\rho).$$ 
If $t$ is a $\beta\eta$-conversion from $\lambda$-term $M$ to $N$, by Definition \ref{Definition-Homotopy-Type-Free-Theory}, $t\in\Lambda_0(M,N)$. For $x,P$ $\lambda$-terms, we have the vertices $\llbracket\lambda x.P \rrbracket_\rho\in K$ and $\llbracket t \rrbracket_{\rho}\in K(\llbracket M \rrbracket_\rho,\llbracket N \rrbracket_\rho)$. Thus, $\llbracket(\lambda x.P)t \rrbracket_\rho=\llbracket\lambda x.P \rrbracket_\rho\bullet\llbracket t \rrbracket_{\rho}\in K(\llbracket(\lambda x. P) M \rrbracket_\rho,\llbracket (\lambda x. P)N \rrbracket_\rho)$ and $\llbracket P \rrbracket_{[\llbracket t \rrbracket_{\rho}/x]\rho}\in K(\llbracket P \rrbracket_{[\llbracket M \rrbracket_{\rho}/x]\rho},\llbracket P \rrbracket_{[\llbracket N \rrbracket_{\rho}/x]\rho})$, where $[\llbracket t \rrbracket_{\rho}/x]\rho:Var\rightarrow K_1$ is an evaluation $\rho'(x)=\llbracket t \rrbracket_{\rho}$ and ($n$-times degeneration of vertex $\rho(r)$) $\rho'(r)=s^n(\rho(r))$  if $r\neq x$. By Definition \ref{Definition-Homotopy-Type-Free-Theory}, $(\lambda x.P)t\in\Lambda_0((\lambda x.P)M,(\lambda x.P)N)$ and $\llbracket P \rrbracket_{[\llbracket t \rrbracket_{\rho}/x]\rho}\in K(\llbracket P \rrbracket_{[\llbracket M \rrbracket_{\rho}/x]\rho},\llbracket P \rrbracket_{[\llbracket N \rrbracket_{\rho}/x]\rho})$. But
$$\llbracket(\lambda x.P)t \rrbracket_\rho\xrightarrow{(\varepsilon_{\llbracket P \rrbracket_{[-/x]\rho}})_{\llbracket t \rrbracket_\rho}}\llbracket P \rrbracket_{[\llbracket t \rrbracket_{\rho}/x]\rho},$$
So $(\lambda x.P)t=[t/x]P$ and induces the $2\beta$-contraction 
$$(\lambda x.P)t \xrightarrow{2\beta_{P,t}}[t/x]P,$$
corresponding to a similar diagram to that of Proposition \ref{proposition1}, i.e.,
	\[\xymatrix{
	& {(\lambda x. P)M}\ar[d]_{(\lambda x. P)t}\ar[r]^{1\beta_{M}} \ar@{}[dr]|{\Longrightarrow_{2\beta_{t}}}&
	[M/x]P \ar[d]^{[t/x]M}\\
	& (\lambda x.P)N\ar[r]_{1\beta_{N}} & [N/x]P
}\] 
Hence $2\beta_{t}\in\Lambda_0((\lambda x.P)M,[N/x]P)(\tau(1\beta_{M},[t/x]M),\tau((\lambda x.P)t,1\beta_{N}))$, where $\tau(r,s)$ is the concatenation of the conversions $r\in \Lambda(a,b)$ and $s\in\Lambda(b,c)$. On the other hand, for $y\notin FV(t)$ one has the equivalence 
$$\llbracket t \rrbracket_\rho\xrightarrow{\eta_{\llbracket t \rrbracket_\rho}}\llbracket \lambda y. ty \rrbracket_\rho,$$
that is, $(\lambda y.ty)=t$ and induces the $2\eta$-contraction
$$(\lambda y.ty)\xrightarrow{2\eta_t}t,$$
which corresponds to the diagram
	\[\xymatrix{
 & \lambda y.My\ar[d]_{\lambda y.ty}\ar[r]^{\hspace{0.3cm}n\eta_r}\ar@{}[dr]|{\Longrightarrow_{2\eta_{t}}} & M\ar[d]^{t}\\
 &\lambda y.Ny\ar[r]_{\hspace{0.3cm}n\eta_s} & N
}\]
In general, if $t\in\Lambda_{n-1}$, the equivalences 
$$\llbracket(\lambda x.P)t \rrbracket_\rho\xrightarrow{(\varepsilon_{\llbracket P \rrbracket_{[-/x]\rho}})_{\llbracket t \rrbracket_\rho}}\llbracket P \rrbracket_{[\llbracket t \rrbracket_{\rho}/x]\rho}, \hspace{0.5cm}\llbracket t \rrbracket_\rho\xrightarrow{\eta_{\llbracket t \rrbracket_\rho}}\llbracket \lambda y. ty \rrbracket_\rho$$ 
in every extensional Kan complex $K$, induce the $(n)\beta\eta$-contractions 
$$(\lambda x.P)t \xrightarrow{n\beta_{t}}[t/x]P, \hspace{0.5cm}(\lambda y.ty)\xrightarrow{n\eta_t}t.$$
which explains the following Corollary. 

\begin{corol}
	If  $x,y,P$ be $\lambda$-terms, $n\geq 1$ and $t\in\Lambda_n(r,s)$ with $y\notin FV(t)$, then the interpretation from diagrams
	\[\xymatrix{
	 &(\lambda x. P)r\ar[d]_{(\lambda x. P)t}\ar[r]^{n\beta_{r}} &
		[r/x]P \ar[d]^{[t/x]M} &&\lambda y.ry\ar[d]_{\lambda y.ty}\ar[r]^{\hspace{0.3cm}n\eta_r} & r\ar[d]^{t}\\
 & (\lambda x. P)s\ar[r]_{n\beta_{s}} & [s/x]P &&\lambda y.sy\ar[r]_{\hspace{0.3cm}n\eta_s} & s
	}\]
	commutes in every extensional Kan complex $K$.	
\end{corol}

Thus, any reflexive Kan complex inductively induces, for each $n\geq 1$, from an $(n)\beta\eta$-conversion $t$ to the $(n+1)\beta\eta$-contractions

	\[\xymatrix{
	&(\lambda x. P)r\ar[d]_{(\lambda x. P)t}\ar[r]^{n\beta_{r}} \ar@{}[dr]|{\Longrightarrow_{(n+1)\beta_{t}}}&
	[r/x]P \ar[d]^{[t/x]M} &&\lambda y.ry\ar[d]_{\lambda y.ty}\ar[r]^{\hspace{0.3cm}n\eta_r}\ar@{}[dr]|{\Longrightarrow_{(n+1)\eta_{t}}} & r\ar[d]^{t}\\
	& (\lambda x. P)s\ar[r]_{n\beta_{s}} & [s/x]P &&\lambda y.sy\ar[r]_{\hspace{0.3cm}n\eta_s} & s
}\]
and these, in their turn, define the $(n+1)\beta\eta$-conversions, of $(n)\beta\eta$-conversion, which would inhabit the set $\Lambda_{n+1}$.

\section{Extensional Kan complexes and Identity types based on higher $\lambda$-terms}

In this section, we use the extensionality of any extensional Kan complex $K$ to define the set of $\lambda^{n}$-terms $\Lambda^{n-1}(a,b)$ induced by the space $K_{n-1}(\llbracket a \rrbracket_\rho,\llbracket b \rrbracket_\rho)$, which would be a type-free version of the identity type $Id_A(a,b)$ based on computational paths of \cite{Queiroz}. And finally we see the relationship between the set $\Lambda^n$ of all the $\lambda^n$-terms and the set $\Lambda_n$ from the previous section.

\medskip By Definition of Cartesian product of simplicial sets one has that for each $n\geq 0$, $(K\times K)_n=K_n\times K_n$. If $\mathcal{K}=\langle K,F,G,\varepsilon,\eta\rangle$ is an extensional Kan complex, then $K_n\times K_n\simeq K_n$, that is 
$K_n\simeq [K_n\rightarrow K_n]$. Hence  $\mathcal{K}_n=\langle K_n,F,G,\varepsilon,\eta\rangle$ is an extensional Kan complex for each $n\geq 0$. 

\medskip For example the case $n=1$, one has that $\llbracket 1\beta \rrbracket_\rho,\llbracket 1\eta \rrbracket_\rho\in K_1$, that is $1\beta,1\eta$ would be `$\lambda^1$-terms'. Hence, for any  $\beta\eta$-conversion $r$ between $\lambda$-terms,  $\llbracket r \rrbracket_\rho\in K_1$, i.e., $r$ would be also a `$\lambda^1$-term' (denoted by $r\in\Lambda^1$). If $h(r)$ is a $\beta\eta$-conversion which depends on the $\beta\eta$-conversion $r$, by extensionality of $K_1$, one has 
$$\llbracket \lambda^1r.h(r) \rrbracket_\rho:=G(\llbracket h(r) \rrbracket_{[-/r ]_\rho})\in K_1,$$
 where  $\llbracket h(r) \rrbracket_{[-/r]\rho}:K_1\rightarrow K_1$.

\medskip Thus, for $m,r\in \Lambda^0(c,d)$ ($\lambda^1$-terms from $c$ to $d$) the `$\lambda^1$-term' $\lambda^1r.h(r)$ can define the $\beta_{2}$-contraction
$$(\lambda^1r.h(r))m\xrightarrow{\beta_{2}} h(m/r)$$
where 
$$\llbracket (\lambda^1r.h(r))m \rrbracket_\rho:=\llbracket \lambda^1r.h(r) \rrbracket_\rho\bullet_{\Delta^1}\llbracket m \rrbracket_\rho=F(\llbracket \lambda^1r.h(r) \rrbracket_\rho)(\llbracket m \rrbracket_\rho)\in K_1,$$
hence, $(\lambda^1r.h(r))m$ can be seen as a $\lambda^1$-term. 

\medskip The question arises:  $\llbracket \beta_2\rrbracket_\rho\in K_2$? To answer this question, let us first prove the following proposition.

\begin{prop}\label{proposition5}
	Let $\mathcal{K}=\langle K, F,G,\varepsilon,\eta\rangle$ be an extensional Kan complex. For each vertex $a,b,c,d\in K$ one has an equivalence of homotopy
	$$K(a,b)\simeq [K(c,d)\rightarrow K(a\bullet c,b\bullet d)],$$
	and in general, for $n\geq 1$ and the vertices $a_{i+1},b_{i+1}\in K(a_0,b_0)\cdots(a_i,b_i)$ and $c_{i+1},d_{i+1}\in K(c_0,d_0)\cdots(c_i,d_i)$ with $0\leq i\leq n-1$, there is an equivalence
	$$K(a_0,b_0)\cdots (a_n,b_n)\simeq [K(c_0,d_0)\cdots (c_n,d_n)\rightarrow K(a_0\bullet c_0,b_0\bullet d_0)\cdots (a_n\bullet c_n,b_n\bullet d_n)]$$
	\begin{proof}
		Since $\mathcal{K}$ is extensional, there is the equivalence $F':K\times K\rightarrow K$. Hence
		$$K(a,b)\times K(c,d)=(K\times K)((a,c),(b,d))\simeq K(F'(a,c),F'(b,d)),$$
		that is, 
		$$K(a,b)\simeq [K(c,d)\rightarrow K(F(a)(c),F(b)(d))]=[K(c,d)\rightarrow K(a\bullet c,b\bullet d)].$$
		Let $K_n(p_n,q_n)=K(p_0,q_0)\cdots(p_n,q_n)$ for each $p_i,q_i\in K_i$ with $0\leq i\leq n$. Given the Induction Hypothesis (IH)
		$$K_n(a_n,b_n)\times K_n(c_n,d_n)\simeq K_n(F'(a_n,c_n),F'(b_n,d_n)),$$
		for the case $(n+1)$ one has
		\begin{align*}
			&K_{n+1}(a_{n+1},b_{n+1})\times K_{n+1}(c_{n+1},d_{n+1})= \\
			&=K_{n}(a_n,b_n)(a_{n+1},b_{n+1})\times K_{n}(c_n,d_n)(c_{n+1},d_{n+1}) \\
			&=(K_{n}(a_n,b_n)\times K_{n}(c_n,d_n))((a_{n+1},c_{n+1}),(b_{n+1},d_{n+1})) \\
			&\simeq K_n(F'(a_n,c_n),F'(b_n,d_n))(F'(a_{n+1},c_{n+1}),F'(b_{n+1},d_{n+1}))\hspace{0.5cm}\text{(by I.H)} \\
			&=K_{n+1}(F'(a_{n+1},c_{n+1}),F'(b_{n+1},d_{n+1})).
			\end{align*} 
Thus,
\begin{align*}
K_{n+1}(a_{n+1},b_{n+1})&\simeq [K_{n+1}(c_{n+1},d_{n+1})\rightarrow K_{n+1}(F(a_{n+1})(c_{n+1}),F(b_{n+1})(d_{n+1}))] \\
&= [K_{n+1}(c_{n+1},d_{n+1})\rightarrow K_{n+1}(a_{n+1}\bullet c_{n+1},b_{n+1}\bullet d_{n+1})]. 
\end{align*}

	\end{proof}
\end{prop}

Therefore, the Proposition \ref{proposition5} allows the following definition.

\begin{defin}\label{paths-computation-interpretation-definition}
	Let $\mathcal{K}=\langle K, F,G,\varepsilon,\eta\rangle$ be an extensional Kan complex and $\rho$ be a valuation in $K$. For the $\beta\eta$-conversions $r,s,h(r)$ such that $\llbracket r \rrbracket_\rho\in K(c,d)$, $\llbracket s \rrbracket_\rho\in K(a,b)$ and $\llbracket h(r) \rrbracket_\rho\in K(a\bullet c,b\bullet d)$, define the interpretation by induction as follows
	\begin{enumerate}
		\item $\llbracket r \rrbracket_\rho\in K(c,d)$ is a concatenation of morphisms
         $$c\xrightarrow{f_1}c_1\xrightarrow{f_2}c_2\xrightarrow{f_3}\cdots\xrightarrow{f_m}d$$		
		 where each $f_i$ depends on: $(\varepsilon_g)_{a}:F(G(g))(a)\rightarrow g(a)$ (interprets each $\beta$-contraction of $r$) or $\eta_{b}:b\rightarrow G(F(b))$ (interprets each inverted $\eta$-contraction of $r$),  with $g\in [K\rightarrow K]$ and $a,b\in K$,
		\item $\llbracket sr \rrbracket_\rho=\llbracket s \rrbracket_\rho\bullet_{\Delta^1}\llbracket r \rrbracket_\rho=F(\llbracket s \rrbracket_\rho)(\llbracket r \rrbracket_\rho)\in K(a\bullet c,b\bullet d)$,
		\item $\llbracket \lambda^1 r.h(r) \rrbracket_\rho=G(\llbracket h(r) \rrbracket_{[-/ r]\rho})\in K(a,b)$ where  $\llbracket h(r) \rrbracket_{[-/ r]\rho}:K(c,d)\rightarrow K(a\bullet c,b\bullet d)$. 
		
		Take $n\geq 2$. For the $(\beta\eta)_n$-conversions (Definition \ref{Definition-Higher-equaity-Theory}) $r$, $s$ and $h(r)$ such that $\llbracket r \rrbracket_\rho\in K_{n-1}(c_{n-1},d_{n-1})$, $\llbracket s \rrbracket_\rho\in K_{n-1}(a_{n-1},b_{n-1})$  and $\llbracket h(r) \rrbracket_\rho\in K_{n-1}(a_{n-1}\bullet c_{n-1},b_{n-1}\bullet d_{n-1})$, define the interpretation
		\item $\llbracket r \rrbracket_\rho\in K_{n-1}(c_{n-1},d_{n-1})$ is a concatenation of $n$-simplexes 
		$$c_{n-1}\xrightarrow{f_1}s_1\xrightarrow{f_2}s_2\xrightarrow{f_3}\cdots\xrightarrow{f_m}d_{n-1}$$		
		 where each $f_i$ depends on: $(\varepsilon_g)_{e}:F(G(g))(e)\rightarrow g(e)$ (interprets each $\beta_n$-contraction of $r$) or $\eta_{e'}:e'\rightarrow G(F(e'))$ (interprets each inverted $\eta_n$-contraction of $r$),  with 
		
		$g:K_{n-1}(c_{n-1},d_{n-1})\rightarrow K_{n-1}(a_{n-1}\bullet c_{n-1},b_{n-1}\bullet d_{n-1})$, 
		$e\in K_{n-1}(c_{n-1},d_{n-1})$ and 	$e'\in K_{n-1}(a_{n-1},b_{n-1})$,
		\item $\llbracket sr \rrbracket_\rho=\llbracket s \rrbracket_\rho\bullet_{\Delta^n}\llbracket r \rrbracket_\rho=F(\llbracket s \rrbracket_\rho)(\llbracket r \rrbracket_\rho)\in K(a\bullet c,b\bullet d)$,
		\item $\llbracket \lambda^n r.h(r) \rrbracket_\rho=G(\llbracket h(r) \rrbracket_{[-/r]_\rho})\in K(a,b)$ where  
		
		$\llbracket h(r) \rrbracket_{[-/r]_\rho}:K_{n-1}(c_{n-1},d_{n-1})\rightarrow K_{n-1}(a_{n-1}\bullet c_{n-1},b_{n-1}\bullet d_{n-1})$.
		 
	\end{enumerate}
\end{defin}

Going back to the question: $\llbracket \beta_2\rrbracket_\rho\in K_2$? Since $\llbracket \lambda^1r.h(r) \rrbracket_\rho\in K_1$, so there are vertices $a,b\in K$ such that $\llbracket \lambda^1r.h(r) \rrbracket_\rho\in K(a,b)$. If $\llbracket r\rrbracket_\rho,\llbracket m\rrbracket_\rho\in K(c,d)$, by Definition \ref{paths-computation-interpretation-definition} (2), $\llbracket (\lambda^1r.h(r)) m\rrbracket_\rho,\llbracket h(m/r)\rrbracket_\rho\in K(a\bullet c,b\bullet d)$. Hence, 
$$\llbracket \beta_2\rrbracket_\rho\in K(a\bullet c,b\bullet d)(a_1,b_1)\subseteq K_2,$$
where $a_1=\llbracket (\lambda^1r.h(r)) m\rrbracket_\rho$ and $b_1=\llbracket h(m/r)\rrbracket_\rho$. 

\medskip For the question: $ \llbracket \eta_2\rrbracket_\rho\in K_2$? Let $e\in K(a,b)$ which does not depend on $r\in K(c,d)$. By Definition \ref{paths-computation-interpretation-definition} (2), $\llbracket er\rrbracket_\rho\in K(a\bullet c,b\bullet d)$. By Definition \ref{paths-computation-interpretation-definition} (3), 
$\llbracket \lambda^1r.er\rrbracket_\rho\in K(a,b)$. Then,
$$\llbracket \eta_2\rrbracket_\rho \in K(a,b)(a_1,b_1)\subseteq K_2,$$
where $a_1=\llbracket \lambda^1r.er\rrbracket_\rho$ and $b_1=\llbracket e\rrbracket_\rho$.

\medskip Therefore, the $(\beta\eta)_2$-conversions are $\lambda^2$-terms, which in turn define inductively other $\lambda^2$-terms by application and abstraction. We can continue iterating and have the following proposition, to prove that the Definition \ref{paths-computation-interpretation-definition} (4) is well defined for all $n\geq 2$.
	
\begin{prop}\label{proposition6}
Let $K$ be an extensional Kan complex  and $\rho:Var\rightarrow K$ be an evaluation. For each $n\geq 1$, $\llbracket \beta_n\rrbracket_\rho,\llbracket \eta_n\rrbracket_\rho\in K_n$.	
\end{prop}
\begin{proof}
	If $n=1$, one has that 
	$\llbracket \beta_1\rrbracket_\rho=\llbracket 1\beta\rrbracket_\rho\in K_1$ and $\llbracket \eta_1\rrbracket_\rho=\llbracket 1\eta\rrbracket_\rho\in K_1$.
	Suppose that $\llbracket \beta_n\rrbracket_\rho,\llbracket \eta_n\rrbracket_\rho\in K_n$. So, induce the $\lambda^n$-terms: $r,m\in \Lambda_{n-1}(c_{n-1},d_{n-1})$ and $\lambda^n r.h(r)\in\Lambda_{n-1}(a_{n-1},b_{n-1})$. By Proposition \ref{proposition5} and Definition \ref{paths-computation-interpretation-definition} (5),
	
	$\llbracket (\lambda^n r.h(r))m \rrbracket_\rho,\llbracket h(m/r) \rrbracket_\rho\in K_{n-1}(a_{n-1}\bullet c_{n-1},b_{n-1}\bullet d_{n-1})$. Thus,
	$$\llbracket \beta_{n+1}\rrbracket_\rho\in K_{n-1}(a_{n-1}\bullet c_{n-1},b_{n-1}\bullet d_{n-1})(a_n,b_n)\subseteq K_{n+1},$$
	where $a_{n}=\llbracket (\lambda^nr.h(r)) m\rrbracket_\rho$ and $b_n=\llbracket h(m/r)\rrbracket_\rho$.  
	
	\medspace By I.H, let the $\lambda^n$-term: $\llbracket e\rrbracket_\rho\in K_{n-1}(a_{n-1},b_{n-1})$ which does not depend on $\llbracket r\rrbracket_\rho\in K_{n-1}(c_{n-1},d_{n-1})$. By Definition \ref{paths-computation-interpretation-definition} (5), $\llbracket er\rrbracket_\rho\in K_{n-1}(a_{n-1}\bullet c_{n-1},b_{n-1}\bullet d_{n-1})$. By Definition \ref{paths-computation-interpretation-definition} (6), 
	$\llbracket \lambda^nr.er\rrbracket_\rho\in K_{n-1}(a_{n-1},b_{n-1})$. So,
	$$\llbracket \eta_{n+1}\rrbracket_\rho \in K_{n-1}(a_{n-1},b_{n-1})(a_n,b_n)\subseteq K_{n+1},$$
	where $a_n=\llbracket \lambda^nr.er\rrbracket_\rho$ and $b_n=\llbracket e\rrbracket_\rho$.
\end{proof}

Of course, Definition \ref{paths-computation-interpretation-definition} depends on the syntax of higher lambda-terms. Next, we define a `Theory of higher $\lambda\beta\eta$-equality' as a type-free version of the computational paths of \cite{Queiroz}.

\begin{defin}[Theory of higher $\lambda\beta\eta$-equality]\label{Definition-Higher-equaity-Theory}
	A theory of higher $\lambda\beta\eta$-equality (TH-$\lambda\beta\eta$) consists of rules and axioms of the theory of $\beta\eta$-equality  ($\beta\eta$-conversions or in our case we write $(\beta\eta)_1$-conversions) between $\lambda$-terms, whose set we denote here by $\Lambda^0$, and the rules which define the higher $\beta\eta$-conversions  in the following sense:
	\begin{itemize}
		\item (1-introduction and 1-formation rules). $s$ is a $(\beta\eta)_1$-conversion from $\lambda$-term $a$ to $\lambda$-term $b$ (denoted by $a=_sb\in\Lambda^0$) if $s$ is a usual $(\beta\eta)$-conversion from $a$ to $b$, and we say that all $(\beta\eta)_1$-conversion  is a $\lambda^1$-term.
		
		\medskip Let $c=_md\in\Lambda^0$ and $[c=_rd\in\Lambda^0]$ $ac=_{h(r)}bd\in\Lambda_0$. Then $\lambda^{1}r.h(r)$ is a $\lambda^1$-term from $a$ to $b$, i.e., $\lambda^{1}r.h(r)\in\Lambda^0(a,b)$ and $(\lambda^{1}r.h(r))m$ is a $\lambda^1$-term from $ac$ to $bd$, i.e., $(\lambda^{1}r.h(r))m\in \Lambda^0(ac,bd)$. Let $\Lambda^1$ the set of the $\lambda^1$-terms.

		\item (Reduction rule). Let the $\lambda^{n+1}$-terms $m\in\Lambda^n(c,d)$, $[r\in\Lambda^n(c,d)]$ and $h(r)\in\Lambda^n(ac,ad)$. Define the $\lambda^{n+1}$-term: $\lambda^{n+1}r.h(r)\in\Lambda^{n}(a,b)$ and the  $\beta_{n+2}$-contraction
		$$(\lambda^{n+1}r.h(r))m\xrightarrow{\beta_{n+2}}h(m/r)\in\Lambda^n(ac,bd).$$
		\item  (Induction rule). If $t\in\Lambda^n(c,d)$ and  $e\in\Lambda^n(a,b)$, then  $\eta_{n+2}$-contraction is given by 
		$$\lambda^{n+1}t.et\xrightarrow{\eta_{n+2}}e\in\Lambda^n(a,b),$$
		where $e$ does not depend on $t$.
		\item ($(n+2)$-Introduction and $(n+2)$-formation rules). If $s$ is a $(\beta\eta)_{n+2}$-conversion (sequence, it can be empty, of $\beta_{n+2}$-contractions or reversed $\beta_{n+2}$-contractions or $\eta_{n+2}$-contractions or reversed $\eta_{n+2}$-contractions) from  $a$ to $b$ in $\Lambda^{n+1}$, that is  $a=_sb\in\Lambda^{n+1}$, then $s\in \Lambda^{n+1}(a,b)$. We say that $s$ is a $\lambda^{n+2}$-term if it is a $(\beta\eta)_{n+2}$-conversion. 
		
		\medskip Let $m\in\Lambda^{n+1}(c,d)$ and $[c=_rd\in\Lambda^{n+1}]$. Then one has the  $\lambda^{n+2}$-terms: $\lambda^{n+2}r.h(r)\in\Lambda^{n+1}(a,b)$ and
		$(\lambda^{n+2}r.h(r))m\in\Lambda^{n+1}(ac,bd).$ Let $\Lambda^{n+2}$ be the set of the $\lambda^{n+2}$-terms.
	\end{itemize}  
\end{defin}

\begin{prop}\label{Proposition4}
	Let $\mathcal{K}=\langle K,F,G,\varepsilon,\eta\rangle$ be an extensional Kan complex and  $\rho:Var\rightarrow K$ be an evaluation. The $(n+1)$-simplexes space $K_n(\llbracket p \rrbracket_\rho,\llbracket q \rrbracket_\rho)$  models the set of $\lambda^{n+1}$-terms $\Lambda^n(p,q)$.
\end{prop}
\begin{proof}
\begin{itemize}
		\item (1-Formation and 1-introduction rules). Since $K$ is a Kan complex and $p,q\in\Lambda^0$, then $\llbracket p \rrbracket_\rho,\llbracket q \rrbracket_\rho\in K$ (vertices of $K$) and $K(p,q)$ is also a Kan complex.  
		
		Let $p=_sq\in\Lambda^0$ be a $(\beta\eta)_1$-conversion. Since $K$ is an extensional Kan complex, by Definition \ref{paths-computation-interpretation-definition} the interpretation
		$$\llbracket s \rrbracket_\rho:\llbracket p \rrbracket_\rho\xrightarrow{f_1}\llbracket p^1 \rrbracket_\rho\xrightarrow{f_2}\llbracket p^2 \rrbracket_\rho\xrightarrow{f_3}\cdots\xrightarrow{f_m}\llbracket q \rrbracket_\rho$$ 
		is a concatenation of morphisms in $K$ such that each $f_i$ corresponds to a morphism which depends on a map of the form: $(\varepsilon_g)_a:F(G(g))(a)\rightarrow g(a)$ (models the $\beta_1$-contraction)  or  $\eta_b:b\rightarrow G(F(b))$ (models the reversed $\eta_1$-contraction), where  $a,b\in K$ and $g\in[K\rightarrow K]$. Thus $\llbracket s\rrbracket_\rho\in K(\llbracket p \rrbracket_\rho,\llbracket q \rrbracket_\rho)$.
		
		\medskip Let $m\in\Lambda^0(s,t)$ and $[s=_rt\in\Lambda^0]$ $\lambda^{1}r.h(r)\in \Lambda^0(p,q)$. Since $\mathcal{K}$ is extensional,  by Definition \ref{paths-computation-interpretation-definition}
		$$\llbracket(\lambda^{1}r.h(r))m\rrbracket_\rho=F(G(\llbracket h(r) \rrbracket_{[-/ r]\rho}))(\llbracket m \rrbracket_\rho)\in K(\llbracket ps \rrbracket_\rho,\llbracket qt \rrbracket_\rho).$$
		
		\item (Reduction rule). Let $m\in\Lambda^n(s,t)$ and $[s=_rt\in\Lambda^n]$ $\lambda^{n+1}r.h(r)\in\Lambda^n(p,q)$. Since $K$ is extensional, the $\beta_{n+2}$-contraction
		$$(\lambda^{n+1}r.h(r))m\xrightarrow{\beta_{n+2}}h(m/r)\in\Lambda^n(ps,qt)$$
		corresponds to morphism in $K_n(\llbracket ps \rrbracket_\rho,\llbracket qt \rrbracket_\rho)$ ($(n+2)$-simplex at $K$): 
		 $$F(G(\llbracket h(r) \rrbracket_{[-/ r]\rho}))(\llbracket m \rrbracket_\rho)\xrightarrow{(\varepsilon_{\llbracket h(r) \rrbracket_{[-/ r]\rho}})_{\llbracket m \rrbracket_\rho}}\llbracket h(m/r) \rrbracket_\rho.$$
		\item (Induction rule). Let $r\in\Lambda^n(p,q)$ and  $e\in\Lambda^n(p,q)$. Since $K$ is extensional, the $\eta_{n+2}$-contraction
		$$\lambda^{n+1} t.et\xrightarrow{\eta_{n+2}}e\in\Lambda^n(p,q)$$
		corresponds to morphism in $K_n(\llbracket p \rrbracket_\rho,\llbracket q \rrbracket_\rho)$:
		$$G(F(\llbracket e \rrbracket_\rho))\xrightarrow{\tilde{\eta}_{\llbracket e \rrbracket_\rho}}\llbracket e \rrbracket_\rho,$$
		where $\tilde{\eta}_{\llbracket e \rrbracket_\rho}$ is an inverse (up to homotopy) from $(n+2)$-simplex $\eta_{\llbracket e \rrbracket_\rho}$ in $K$. 
		\item ($(n+2)$-Introduction and $(n+2)$-Formation rules). Take the $(\beta\eta)_{n+2}$-conversion $s=_rt\in\Lambda^{n+1}$. Since $K$ is an extensional Kan complex, by Definition \ref{paths-computation-interpretation-definition} the interpretation
		$$\llbracket r \rrbracket_\rho:\llbracket s \rrbracket_\rho\xrightarrow{f_1}\llbracket s^1 \rrbracket_\rho\xrightarrow{f_2}\llbracket s^2 \rrbracket_\rho\xrightarrow{f_3}\cdots\xrightarrow{f_m}\llbracket t \rrbracket_\rho$$ 
		is a concatenation of morphisms in $K_{n+1}$ such that each $f_i$ corresponds to a morphism which depends on a map of the form: $(\varepsilon_g)_e:F(G(g))(e)\rightarrow g(e)$ (models the $\beta_{n+2}$-contraction)  or  $\eta_{e'}:e'\rightarrow G(F(e'))$ (models the reversed $\eta_{n+2}$-contraction), where  $e\in K_{n+1}(c_n,d_n)$, $e'\in K_{n+1}(a_n,b_n)$ and $g:K_{n+1}(c_n,d_n)\rightarrow K_{n+1}(a_n\bullet c_n,b_n\bullet d_n)$. Thus $\llbracket r(s,t) \rrbracket_\rho\in K_{n+1}(\llbracket s \rrbracket_\rho,\llbracket t \rrbracket_\rho)$.
		
		\medskip Let $m\in\Lambda^{n+1}(s,t)$ and $[s=_rt:A]$ $\lambda^{n+2}r.h(r)\in\Lambda^{n+1}(p,q)$. Since $\mathcal{K}$ is extensional,  by Definition \ref{paths-computation-interpretation-definition}
		$$\llbracket(\lambda^{n+2}r.h(r))m\rrbracket_\rho=F(G(\llbracket h(r) \rrbracket_{[-/r ]_\rho}))(\llbracket m \rrbracket_\rho)\in K_{n+1}(\llbracket ps \rrbracket_\rho,\llbracket qt \rrbracket_\rho).$$
	\end{itemize}
\end{proof}

\begin{ejem}
Let $c=_md\in\Lambda^0$ and $[c=_rd\in\Lambda^0]$ $ac=_{h(r)}bd\in\Lambda^0$, thus $\lambda^1 r. h(r)\in\Lambda^0(a,b)$. The $\beta_2$-contraction is $2$-dimensional. It can be represented by the diagram
	\[\xymatrix{
	& {ac}\ar[d]_{(\lambda^1 r. h(r))m}\ar[r]^{1} \ar@{}[dr]|{\Longrightarrow_{\beta_{2}}}&
	ac \ar[d]^{h(m/r)}\\
	& bd\ar[r]_{1} & bd
}\] 
Since the interpretation of $\lambda^1 r. h(r)\in\Lambda^0(a,b)$ is given by 
$$\llbracket \lambda^1 r.h(r) \rrbracket_\rho=G(\llbracket h(r) \rrbracket_{[-/ r]\rho})\in K(\llbracket a\rrbracket_\rho,\llbracket b\rrbracket_\rho)$$
for every extensional Kan complex $\mathcal{K}$ and $\rho$, by Definition \ref{Definition-Homotopy-Type-Free-Theory} one has$ \lambda^1 r.h(r) \in \Lambda_0(a,b)$. And the interpretation of the application $\lambda^1 r. h(r))m$ is given by
$$\llbracket (\lambda^1 r.h(r))m \rrbracket_\rho=\llbracket \lambda^1 r.h(r) \rrbracket_\rho\bullet_{\Delta^1}\llbracket m \rrbracket_\rho=F(\llbracket \lambda^1 r.h(r) \rrbracket_\rho)(\llbracket m \rrbracket_\rho)\in K(\llbracket ac\rrbracket_\rho,\llbracket bd\rrbracket_\rho)$$
for all extensional Kan complex $\mathcal{K}$ and $\rho$. By Definition \ref{Definition-Homotopy-Type-Free-Theory} $(\lambda^1 r.h(r))m\in\Lambda_0(ac,bd)$.  Therefore $\Lambda^1=\Lambda_1$. 

Follow the question:  $\beta_2\in\Lambda_2$?  By Proposition \ref{Proposition4} (Reduction rule for $n=0$) the $\beta_2$-contraction is interpreted by the 2-simplex
	 $$F(G(\llbracket h(r) \rrbracket_{[-/ r]\rho}))(\llbracket m \rrbracket_\rho)\xrightarrow{(\varepsilon_{\llbracket h(r) \rrbracket_{[-/ r]\rho}})_{\llbracket m \rrbracket_\rho}}\llbracket h(m/r) \rrbracket_\rho\in K(\llbracket ac\rrbracket_\rho,\llbracket bd\rrbracket_\rho)$$
	 
\noindent for all extensional Kan complex $\mathcal{K}$ and evaluation $\rho$. By Definition \ref{Definition-Homotopy-Type-Free-Theory} one has $\beta_2\in\Lambda_0(ac,bd)(\lambda^1 r.h(r))m,h(m/r))$. Hence $\beta_2\in\Lambda_2$. 
\end{ejem}

One the other hand, by Proposition \ref{Proposition4} (Induction rule for $n=0$) and the same reasoning from previous example, it can be proved that $\eta_2\in\Lambda_2$, so $\Lambda^2\subseteq\Lambda_2$. Thus making use of Definitions \ref{interpretation-definition} and \ref{Definition-Homotopy-Type-Free-Theory} and  Proposition \ref{Proposition4} we can prove in the same way as the previous example, the following proposition.

\begin{prop}
	For each $n\geq 0$, $\Lambda^n\subseteq\Lambda_n$. Hence TH-$\lambda\beta\eta\subseteq HoTFT$.
\end{prop}

\section{Conclusion}

We define the interpretation of the $\beta\eta$-contractions in an extensional Kan complex, whose $\infty$-groupoid structure induces higher  $\beta\eta$-contractions, which consolidate a type-free version of HoTT, which we call HoTFT (Homotopy Type-Free Theory), which could have the advantage of rescuing the $\beta\eta$-conversions as relations of intentional equality and not as relations of judgmental equality as is the case in HoTT.

\bigskip Besides, we define, from the identity types based on computational paths, the untyped theory of higher $\lambda\beta\eta$-equality  TH-$\lambda\beta\eta$, which is contained in HoTFT.

\bibliography{mybibfile}

\begin{thebibliography}{1}
\expandafter\ifx\csname url\endcsname\relax
  \def\url#1{\texttt{#1}}\fi
\expandafter\ifx\csname urlprefix\endcsname\relax\def\urlprefix{URL }\fi
\expandafter\ifx\csname href\endcsname\relax
  \def\href#1#2{#2} \def\path#1{#1}\fi

\bibitem{Martinez}
D.~Mart\'{\i}nez-Rivillas, R.~de~Queiroz,
  \href{https://doi.org/10.1093/jigpal/jzab015}{The $\infty$-groupoid generated
  by an arbitrary topological $\lambda$-model}, Logic Journal of the IGPL.
  (also arXiv:1906.05729) 30 (2022) 465--488.
\newline\urlprefix\url{https://doi.org/10.1093/jigpal/jzab015}

\bibitem{MartinezHoDT}
D.~Mart\'{\i}nez-Rivillas, R.~de~Queiroz,
  \href{https://doi.org/10.1007/s00153-022-00856-0}{Towards a homotopy domain
  theory}, Archive for Mathematical Logic (also arXiv 2007.15082) (2022).
\newline\urlprefix\url{https://doi.org/10.1007/s00153-022-00856-0}

\bibitem{Martinez21}
D.~Mart\'{\i}nez-Rivillas, R.~de~Queiroz, Solving homotopy domain equations,
  arXiv:2104.01195 (2021).

\bibitem{Queiroz}
R.~de~Queiroz, A.~de~Oliveira, A.~Ramos, Propositional equality, identity
  types, and direct computational paths, South American Journal of Logic 2~(2)
  (2016) 245--296.

\bibitem{DBLP:books/mk/Lurie}
J.~Lurie, Higher Topos Theory, Princeton University Press, Princeton and
  Oxford, 2009.

\bibitem{DBLP:books/mk/Rezk22}
C.~Rezk,
  \href{https://faculty.math.illinois.edu/\~{}rezk/quasicats.pdf}{Introduction
  to Quasicategories}, Lecture Notes for course at University of Illinois at
  Urbana-Champaign, 2022.
\newline\urlprefix\url{https://faculty.math.illinois.edu/\~{}rezk/quasicats.pdf}

\end{thebibliography}

\end{document}